\documentclass[final]{l4dc2025}

\title[Deception in Stochastic NES]{Stochastic Real-Time Deception in Nash Equilibrium Seeking for Games with Quadratic Payoffs}
\usepackage{times}
\usepackage{amsmath,amssymb,amsfonts}
\usepackage{graphicx}
\usepackage{mathtools}
\usepackage{float}

\newtheorem{assumption}{Assumption}
\newcommand\ca{\mathcal{A}}
\newcommand\cb{\mathcal{B}}
\newcommand\ck{\mathcal{K}}

\author{%
 \Name{Michael Tang} \Email{myt001@ucsd.edu}\\
 \addr Department of Electrical and Computer Engineering, University of California, San Diego\\
 \Name{Miroslav Krstic} \Email{mkrstic@ucsd.edu}\\
 \addr Department of Mechanical and Aerospace Engineering, University of California, San Diego\\
 \Name{Jorge Poveda} \Email{poveda@ucsd.edu}\\
 \addr Department of Electrical and Computer Engineering, University of California, San Diego%
\thanks{This work was supported in part by NSF grants ECCS CAREER 2305756, CMMI 2228791, and AFOSR YIP: FA9550-22-1-0211}}

\begin{document}

\maketitle

\begin{abstract}%
 In multi-agent autonomous systems, deception is a fundamental concept which characterizes the exploitation of unbalanced information to mislead victims into choosing oblivious actions. This effectively alters the system's long term behavior, leading to outcomes that may be beneficial to the deceiver but detrimental to victim. We study this phenomenon for a class of model-free Nash equilibrium seeking (NES) where players implement independent stochastic exploration signals to learn the pseudogradient flow. In particular, we show that deceptive players who obtain real-time measurements of other players' stochastic perturbation can incorporate this information into their own NES action update, consequentially steering the overall dynamics to a new operating point that could potentially improve the payoffs of the deceptive players. We consider games with quadratic payoff functions, as this restriction allows us to derive a more explicit formulation of the capabilities of the deceptive players. By leveraging results on multi-input stochastic averaging for dynamical systems, we establish local exponential (in probability) convergence for the proposed deceptive NES dynamics. To illustrate our results, we apply them to a two player quadratic game.
\end{abstract}


\section{Introduction}
\subsection{Motivation}
The study of data-driven multi-agent systems is becoming increasingly relevant, especially in engineering domains such as smart grids, power systems, robotics and machine learning. Game theory offers a powerful framework for analyzing the interactions between agents in cooperative and competitive settings. However, traditional game-theoretic methods mostly assume that all agents share the same amount of information for the purpose of decision-making. When agents gain access to \emph{privileged information}, traditional game-theoretic concepts, such as the Nash equilibrium \cite{540b73bd-a3f1-333e-a206-c24d0fbbb8bc}, may lose applicability. Such asymmetry of information can be exploited to shape the system's long-term behavior in a way that favors certain agents while disadvantaging others. In adversarial environments, some agents may leverage privileged information to mislead their opponents into adopting suboptimal policies. This phenomenon, commonly referred to as \emph{deception in games}, has received significant research interest partly due to the growing need for safe, resilient decision-making systems in game-theoretic and data-driven settings. Deception has been studied in many domains, such as cybersecurity \cite{pawlick2021game}, robotics \cite{9494340} and biological systems \cite{smith1987deception}. While deception can serve as a means by which one achieves their goal in a non-competitive environment, it has also been thoroughly studied in competitive settings, such as signaling games \cite{kouzehgar2019fuzzy}, repeated security games \cite{nguyen2019deception}, etc. Despite the wide variety of applications, a common theme is the asymmetry of information. For an agent to have the capacity to deceive, they must be given access to information that the other agents are unaware of. Considerable research has also been done to identify when to optimally deceive \cite{wagner2011acting}, how to detect deceptive agents \cite{9508449}, and strategies to counter deception \cite{5339236}. 

In this work, we are interested in studying deception within the framework of stochastic model-free Nash equilibrium seeking (NES) algorithms for non-cooperative games. In a non-cooperative game, each player seeks to minimize their cost function, which may depend on the actions of other players. The solution to this optimization problem, known as the \emph{Nash equilibrium (NE)}, is not straightforward to compute. Given this challenge, various NES algorithms have been designed, including distributed \cite{ye2017distributed}, semi-decentralized \cite{belgioioso2017semi} and hybrid \cite{9968117} algorithms. However, in many cases, players do not know the mathematical model of their cost function or the structure of the game. In this case, players must implement adaptive seeking dynamics that incorporate exploration and exploitation strategies using real-time measurements of their cost function values. Such algorithms have been studied in, e.g.,  \cite{6060862,PoQu15,Stipanovik_ES_TAC}, using extremum-seeking ideas. For this class of algorithms, it was recently shown in \cite{tang2024deception} that a player who obtains knowledge of another player's exploration policy can incorporate this information into their own action update, effectively driving the NES dynamics to a point that might be more beneficial for the deceptive player than the nominal NE. Moreover, conditions that guarantee stability of the closed-loop deceptive NES dynamics were derived. While the results raised important concerns regarding the functionality of real-time adaptive learning algorithms in the presence of adversarial behavior, one important limitation is that they rely on the players using deterministic periodic excitation signals to learn the ``deceptive'' Nash equilibrium. However, many practical systems, such as biological systems and source-seeking systems, rely on unpredictable and random probing to learn optimal operating points. To the best of our knowledge, the study of stochastic deception in Nash equilibrium-seeking problems has remained mostly unexplored.
\subsection{Contributions}
In this paper, we introduce the concept of \emph{stochastic deception} in Nash equilibrium seeking problems. In particular, by leveraging averaging results for stochastic differential equations \cite{LiuStochastic}, we extend the deterministic framework of \cite{tang2024deception} to the case where players use random probing for the purpose of exploration. While the type of deception considered in \cite{tang2024deception} only required knowledge of the victim's sinusoidal frequency (which could be interpreted as a ``key" in the learning mechanism), the deception mechanism in this work requires real-time measurements of the victim's stochastic perturbation. In other words, to achieve stochastic deception, the deceiver must be able to establish a sensing link with the victim's probing policy. By focusing on quadratic games with a strictly diagonally dominant pseudogradient, we establish suitable stability properties for the average dynamics, resulting in more explicit characterizations on the influence of deception compared to earlier works. We further generalize the class of deceptive NES algorithms from \cite{tang2024deception} by allowing players to select differing gains and amplitudes for their exploration policies. We then establish conditions which can guarantee stochastic exponential stability of the closed-loop stochastic deceptive NES dynamics. While the deterministic setting required players to choose different exploration frequencies, our work merely requires that each player's perturbation process is chosen independently, resulting in significantly greater flexibility for the players' strategies. All results are further illustrated via numerical examples.

\section{Preliminaries}
Given a matrix $A\in\mathbb{R}^{m\times n}$, we use $[A]_{ij}$ and $[A]_{k:}$ to denote the $(i,j)$ entry and $k$-th row of $A$, respectively. Similarly, given a vector $b\in\mathbb{R}^n$, we use $[b]_i$ to denote the $i$-th entry of $b$. If $N\in\mathbb{N}$, we use $[N]$ to denote the set of positive integers no greater than $N$, i.e $[N]=\{1,2,...,N\}$.
Given a vector $v\in\mathbb{R}^n$, we use $\text{diag}(v)\in\mathbb{R}^{n\times n}$ to denote the diagonal matrix with the 
$i$-th diagonal entry given by $[v]_i$. For $r>0$, we use $\mathcal{B}_r(v)$ to denote the open ball of radius $r$ around $v$, i.e $\mathcal{B}_r(v)=\{p\in\mathbb{R}^n : |p-v|<r\}$, where $|\cdot|$ is the Euclidean norm.
\section{Problem Formulation}
Consider an $N$-player noncooperative game, where player $i\in[N]$ implements action $x_i\in\mathbb{R}$ and aims to unilaterally minimize their cost function $J_i:\mathbb{R}^N\to\mathbb{R}$. We let $x=[x_1,...,x_N]^\top$ denote the vector of all the players' actions, and we use $x_{-i}$ to denote the vector of all players' actions except for that of player $i$. A policy $x^*\in\mathbb{R}^N$ is said to be a \emph{Nash equilibrium (NE)} of the game if it satisfies
\begin{equation}
    J_i(x_i^*, x_{-i}^*)\le J_i(x_i, x_{-i}^*),\quad \forall x_i\in\mathbb{R},\quad\forall i\in[N].
\end{equation}
In this paper, we are primarily interested in games with quadratic payoffs, i.e., we assume that for each player $i$, $J_i$ takes the quadratic form
\begin{equation}\label{jquad}
    J_i(x)=\frac12 x^\top A_i x+b_i^\top x+c_i
\end{equation}
where $A_i\in\mathbb{R}^{N\times N}$ is symmetric, $b_i\in\mathbb{R}^N$, and $c_i\in\mathbb{R}$. With this in mind, the \emph{pseudogradient} of the game can be defined as $\mathcal{G}(x)=\mathcal{A}x+\mathcal{B}$, where $\mathcal{A}\in\mathbb{R}^{N\times N}$ satisfies $[\mathcal{A}]_{i:}=[A_i]_{i:}$ and $\mathcal{B}\in\mathbb{R}^N$ satisfies $[\mathcal{B}]_i=[b_i]_i$ for $i\in[N]$. We make the following assumption on the pseudogradient
\begin{assumption}\label{diag}
    The matrix $\mathcal{A}$ is strictly diagonally dominant and satisfies $[\mathcal{A}]_{ii}>0$ for all $i\in[N]$.
\end{assumption}
For player $i$ to learn their NE strategy using only measurements of their payoff $J_i$, they implement the following stochastic NE seeking algorithm, originally introduced in \cite{doi:10.1137/100811738}:
\begin{subequations}\label{snes}
    \begin{align}
        x_i&=u_i+a_i f_i(\eta_i(t))\\
        \dot{u}_i&=-\frac{\gamma_i k_i}{a_i}f_i(\eta_i(t))J_i(x)
    \end{align}
\end{subequations}
where $a_i, k_i>0$ are gains and $f_i$ is a bounded odd function to be chosen by player $i$ with
\begin{equation}
    \gamma_i=\int_{-\infty}^{\infty} f_i^2(s)\frac{1}{\sqrt{\pi}q_i}e^{-\frac{s^2}{q_i^2}} ds
\end{equation} 
where $q_i>0$. The signals $\eta_i(t), i\in[N]$ are selected to be independent time homogeneous continuous Markov ergodic processes chosen by player $i$, such as the Ornstein-Uhlenbeck (OU) process:
\begin{equation}\label{ou}
    \vartheta_i d\eta_i(t)=-\eta_i(t)dt+\sqrt{\vartheta_i}q_i dW_i(t)
\end{equation}
where $\vartheta_i>0$ and $W_i(t), i\in[N]$ are independent 1-dimensional standard Brownian motion on a complete probability space $(\Omega, \mathcal{F}, P)$ with sample space $\Omega$, $\sigma$-field $\mathcal{F}$, and probability measure $P$. We express $\vartheta_i$ in the form $\vartheta_i=\vartheta \overline{\vartheta}_i$
for some $\vartheta>0$, as this form is convenient for our stability analysis. Without loss of generality, we assume $\overline{\vartheta}_1=1$. Now we introduce the concept of \emph{stochastic deception} within the context of \eqref{snes}:
\begin{definition}
    A player $i$ is said to be \emph{stochastically deceptive} towards a set of players $\mathcal{D}_i\subset [N]\setminus \{i\}$ if its actions are updated according to the following dynamics:
   \begin{subequations}\label{sdnes}
        \begin{align}
            x_i(t)&=u_i(t)+a_i f_i(\eta_i(t))+\delta_i\sum_{k\in\mathcal{D}_i}a_k f_k(\eta_k(t))\label{snesx}\\
            \dot{u}_i&=-\frac{\gamma_i k_i}{a_i}f_i(\eta_i(t))J_i(x)
        \end{align}
    \end{subequations}
    where $\delta_i(t)$ is a tunable deterministic deceptive gain that satisfies $\sup_{t\ge 0} |\delta_i(t)|>0$.
\end{definition}
\begin{figure}[t!]
  \centering
    \includegraphics[width=.95\textwidth]{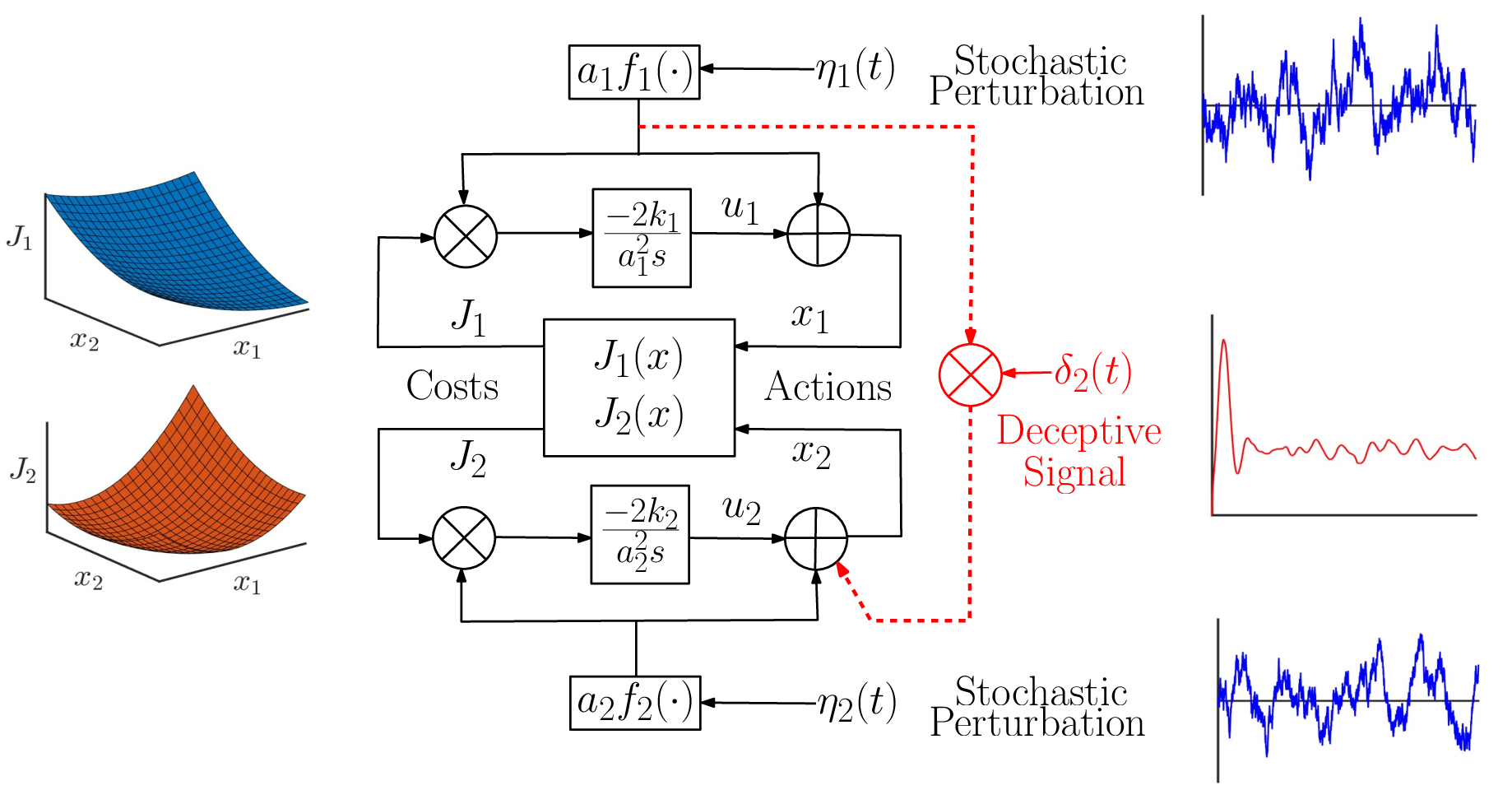}
    \caption{\small Block diagram depicting the deception mechanism for stochastic NES in a two player game where player 2 deceives player 1.}\label{sdnes_block}
    \vspace{-0.5cm}
\end{figure}
For a player to be stochastically deceptive to player $d$, they must obtain knowledge of player $d$'s amplitude $a_d$, their function $f_d$, and real-time measurements of player $d$'s stochastic exploration signal $\eta_d(t)$. While this may seem more restrictive than what is considered in \cite{tang2024deception}, which only required knowledge of exploration frequencies $\omega_d$, we note that our setting is in fact more lenient. Indeed, the NES dynamics in \cite{tang2024deception} assumed $a_k=a$ for all $k\in[N]$ and that all players use a deterministic sinusoidal exploration policy, i.e $f_k(t)=\sin(t)$ and $\eta_k(t)=\omega_k t$ for all $k\in[N]$. Hence, a player who is deceptive to player $d$ in \cite{tang2024deception} can compute $a_d f_d(\eta_d(t))$ in real-time, which is what we require in \eqref{snesx}. Figure \ref{sdnes_block} depicts this phenomenon for a two player game where player 2 is deceptive to player 1.

\vspace{0.1cm}
In this work, we assume there are $n$ deceptive players, and the set of deceptive players is given by $\mathcal{D}=\{d_1,...,d_n\}$. Moreover, we let $\mathcal{V}=\{j\in[N] : \exists i \in [N] \text{ s.t } j\in\mathcal{D}_i\}$ denote the set of ``victims", i.e the players who are being deceived. We say player $i$ is deceiving $n_i\in [N-1]$ players in $\mathcal{D}_i$ if they implement \eqref{sdnes} and update $\delta_i$ according to the following dynamics:
\begin{equation} \dot{\delta}_i=\varepsilon\varepsilon_i (J_i(x)-J_i^\text{ref}).\label{ddelta}
\end{equation}
where $J_i^\text{ref}$ is the desired payoff of player $i$ and $\varepsilon, \varepsilon_i>0$ are small parameters. Now, we can state the full model-free stochastic dynamics for the deceptive game:
\begin{subequations}\label{decgamedyn}
     \begin{align}
     x_i&=\begin{dcases}
        u_i+a_if_i(\eta_i(t))+\delta_i \sum_{j\in\mathcal{D}_i}a_jf_j(\eta_j(t)) & \text{if } i \in \mathcal{D} \\
        u_i+a_i f_i(\eta_i(t)) & \text{else}
    \end{dcases}\label{decgame1}\\
    \dot{u}_i&=-\frac{\gamma_i k_i}{a_i}J_i(x)f_i(\eta_i(t)),\label{decgame2}\\ \dot{\delta}_i&=\varepsilon\varepsilon_i (J_i(x)-J_i^\text{ref}),\quad i\in\mathcal{D}
 \end{align}
\end{subequations}
Let $\delta=[\delta_{d_1},...,\delta_{d_n}]^\top$ and let $\mathcal{K}_j$ represent the set of players who are deceptive to player $j$, i.e., $\mathcal{K}_j=\{i\in\mathcal{D}: j\in\mathcal{D}_i\}$. We use $\zeta\in\mathbb{R}^{N+n}$ to denote the full state vector, i.e $\zeta=[u^\top\quad \delta^\top]^\top$. We define the perturbed pseudogradient parameters as $\overline{\ca}(\delta)\in\mathbb{R}^{N\times N}$ and $\overline{\cb}(\delta)\in\mathbb{R}^N$, which have entries satisfying:
    \begin{align}\label{perturb_pg}
        [\overline{\ca}(\delta)]_{i:}=[\ca]_{i:}+\sum_{k\in\ck_i}\delta_k [A_i]_{k:},\quad [\overline{\cb}(\delta)]_i=[\cb]_i+\sum_{k\in\ck_i}\delta_k [b_i]_k
    \end{align}
The intuition behind \eqref{perturb_pg} will become clear when the averaged dynamics of \eqref{decgamedyn} are computed in Section \ref{sec_proof}. 
We can then define the following \emph{stability preserving set}
\begin{equation}
    \Delta:=\{\delta\in\mathbb{R}^n : -\text{diag}([k_1,...,k_N]^\top)\overline{\ca}(\delta) \text{ is Hurwitz}\}.
\end{equation}
By leveraging Assumption \ref{diag}, we obtain an estimate of $\Delta$ that is somewhat conservative but an improvement of the results from \cite{tang2024deception}.
\begin{lemma}\label{lemma_delta}
    Pick $r\in(0, \infty)$ satisfying
    \begin{equation}\label{deltabd}      r=\min_{i\in\mathcal{V}}\min\left(\frac{[\ca]_{ii}-\sum_{j\neq i}^N|[\ca]_{ij}|}{\sum_{k\in\ck_i}\sum_{j=1}^N |[A_i]_{kj}|}, \frac{[\ca]_{ii}}{\sum_{k\in\ck_i}|[A_i]_{ki}|}\right)
    \end{equation}
    where we interpret $\frac{p}{q}:=\infty$ if $p>0$ and $q=0$. Then $\mathcal{B}_r(\mathbf{0}_n)\subset\Delta$, where $\mathbf{0}_n\in\mathbb{R}^n$ is the $n$-dimensional zero vector.
\end{lemma}
\begin{proof}
    Fix $i\in\mathcal{V}$, and pick $\delta$ satisfying $|\delta|<\frac{[\ca]_{ii}}{\sum_{k\in\ck_i}|[A_i]_{ki}|}$. Then,
    \begin{align}
        [\overline{\ca}(\delta)]_{ii}&=[\ca]_{ii}+\sum_{k\in\ck_i}\delta_k [A_i]_{ki}\ge [\ca]_{ii}-|\delta|\sum_{k\in\ck_i}|[A_i]_{ki}|>0
    \end{align}
    Now let $\delta$ also satisfy $|\delta|<\frac{[\ca]_{ii}-\sum_{j\neq i}^N|[\ca]_{ij}|}{\sum_{k\in\ck_i}\sum_{j=1}^N |[A_i]_{kj}|}$. Then, for $\overline{\ca}(\delta)$ to be strictly diagonally dominant, we require
    \begin{align}    [\ca]_{ii}+\sum_{k\in\ck_i}\delta_k [A_i]_{ki}>\sum_{j\neq i}^N\left\lvert[\ca]_{ij}+\sum_{k\in\ck_i}\delta_k [A_i]_{kj}\right\rvert\label{lemd}
    \end{align}
    For \eqref{lemd} to hold, it suffices to have
    \begin{align}    [\ca]_{ii}-|\delta|\sum_{k\in\ck_i} |[A_i]_{ki}|>\sum_{j\neq i}^N \left(|[\ca]_{ij}|+|\delta|\sum_{k\in\ck_i}|[A_i]_{kj}|\right)
    \end{align}
    which can be guaranteed whenever
    \begin{align}
        [\ca]_{ii}-\sum_{j\neq i}^N|[\ca]_{ij}|>|\delta|\sum_{k\in\ck_i}\sum_{j=1}^N |[A_i]_{kj}|
    \end{align}
    Taking the minimum over all $i\in\mathcal{V}$, it follows that $\overline{\ca}(\delta)$ is strictly diagonally dominant for $|\delta|<r$, and $[\overline{\ca}(\delta)]_{ii}>0$ for all $i\in[N]$. By the Gershgorin Circle Theorem, we conclude that $\delta\in\Delta$.
\end{proof}
This estimate is a significant improvement compared to the results of \cite{tang2024deception}, which only guaranteed that $\Delta$ contains a neighborhood of the origin. By further exploiting the diagonal dominance assumption, we are able to derive an explicit lower bound on the ``size" of this neighborhood. To characterize when the deceptive players are able to achieve their desired payoffs via deception while preserving stability, we introduce the notion of \emph{attainability} as was presented in \cite{tang2024deception}:
\begin{definition}\label{jattaindef}
    A vector $J^{\text{ref}}=[J_{d_1}^{\text{ref}},...,J_{d_n}^{\text{ref}}]^\top$ is said to be attainable if there exists $\delta^*\in\Delta$ such that:
    \begin{enumerate}
        \item $J_{d_k}(-\overline{\ca}(\delta^*)^{-1}\overline{\cb}(\delta^*))=J_{d_k}^{\text{ref}},\quad\forall~~k\in [n]$.
        \item The matrix $\Lambda(\delta^*)\in\mathbb{R}^{n\times n}$ with $[\Lambda(\delta^*)]_{jk}=\nabla_j \xi_k(\delta^*)$ is Hurwitz, where $\xi_k:\mathbb{R}^{n}\to\mathbb{R}$ is given by $\xi_k(\delta):=\varepsilon_{d_k} J_{d_k}(-\overline{\ca}(\delta)^{-1}\overline{\cb}(\delta))$.
    \end{enumerate}
 We let $\mathcal{J}\subset\mathbb{R}^n$ denote the set of all \emph{attainable} vectors $J^{\text{ref}}=[J_{d_1}^{\text{ref}},...,J_{d_n}^{\text{ref}}]^\top$.
\end{definition}
We can now state the main result of this paper, which establishes stochastic stability of the NES dynamics \eqref{decgamedyn}.
\begin{theorem}\label{thm}
    Consider the stochastic NES dynamics \eqref{decgamedyn} where the costs $J_i$ are of the quadratic form \eqref{jquad}, $J^\text{ref}\in\mathcal{J}$, the pseudogradient matrix $\ca$ satisfies Assumption \ref{diag}, the noise $\eta_i(t)$ are generated according to \eqref{ou} and $f_i(\cdot)$ are bounded odd functions for $i\in [N]$. Then for each $\tilde{\varepsilon}>0$ there exists $a^*$ such that for $a_1,...,a_N$ satisfying $\max_i a_i<a^*$, there exists $\varepsilon^*$ such that for all $\varepsilon\in(0, \varepsilon^*)$ there exists $R>0, C>0, M>0$ and $\zeta^*=\begin{bmatrix}
        u^*\\ \delta^*
    \end{bmatrix}\in\mathbb{R}^{N+n}$ satisfying the following:
    \begin{enumerate}
        \item $|J_{d_k}(u^*)-J_{d_k}^\text{ref}|<\tilde{\varepsilon}$ for all $k\in [n]$
        \item For any $\tilde{r}>0$ and any initial condition $\zeta_0\in\mathbb{R}^{N+n}$ with $|\zeta_0-\zeta^*|<R$, the solution of \eqref{decgamedyn} satisfies
        \begin{equation}
            \lim_{\vartheta\to 0}\inf\{t\ge 0 : |\zeta(t)-\zeta^*|>C|\zeta_0-\zeta^*|e^{-Mt}+\tilde{r}\}=\infty\quad \text{a.s.}
        \end{equation}
        \item There exists $\varepsilon_0>0$ and a function $T:(0, \varepsilon_0)\to \mathbb{N}$ such that
        \begin{equation}
            \lim_{\vartheta\to 0} P(|\zeta(t)-\zeta^*|\le C|\zeta_0-\zeta^*|e^{-Mt}+\tilde{r}\quad \forall t\in [0, T(\vartheta)])=1 \quad\text{ with } \lim_{\vartheta\to 0}T(\vartheta)=\infty.
        \end{equation}
    \end{enumerate}
\end{theorem}
This result establishes weak stochastic stability of $\zeta^*$ for \eqref{decgamedyn} under random perturbations generated by \eqref{ou}. Moreover, $\zeta^*$ can be made such that the costs of the deceptive players evaluated at the ``$u$" component (which itself can be made arbitrarily close to the players' actions $x$) are arbitrarily close to their desired reference values. The proof is presented in the next section.
\section{Analysis of the Deceptive NES Dynamics and Proof of Convergence}\label{sec_proof}
To prove Theorem \ref{thm}, we compute the ``average" dynamics of \eqref{decgamedyn} and apply Theorem A.3 from \cite{doi:10.1137/100811738}. First, we let $\eta(t)\in\mathbb{R}^N$ denote the vector of the stochastic components, i.e $\eta(t)=x-u$ where $x$ and $u$ are given by \eqref{decgame1} and \eqref{decgame2}. We can also define $\chi_i(t)=\eta_i(\vartheta_i t)$ and $B_i(t)=\frac{1}{\sqrt{\vartheta_i}}W_i(\vartheta_i t)$. Then, by \eqref{ou} we have
\begin{equation}
    d\chi_i(t)=-\chi_i(t)dt+q_idB_i(t),
\end{equation}
where $[B_1(t),...,B_N(t)]^\top$ is an $N$-dimensional Brownian motion on the probability space $(\Omega, \mathcal{F}, P)$. Following the notation in \cite{doi:10.1137/100811738}, we introduce the change of variable
\begin{equation}
    Z_1(t)=\chi_1(t), \ \ Z_2(t)=\chi_2(t/\overline{\vartheta}_2),..., Z_N(t)=\chi_N(t/\overline{\vartheta}_N),
\end{equation}
which will allow us to rewrite the $\dot{u}_i$ system \eqref{decgame2} in the following way:
\begin{align}\label{dui}
    \dot{u}_i=-\frac{\gamma_i k_i}{a_i}f_i(Z_i(t/\vartheta))J_i(u+\eta(t))
\end{align}
where
\begin{equation}\label{etabar}
    [\eta(t)]_k=\begin{dcases}
        a_kf_k(Z_k(t/\vartheta))+\delta_k \sum_{j\in\mathcal{D}_k}a_jf_j(Z_j(t/\vartheta)) & \text{if } k \in \mathcal{D} \\
        a_k f_k(Z_k(t/\vartheta)) & \text{else}
    \end{dcases}
\end{equation}
Since $(\chi_i(t), t\ge 0)$ is ergodic with invariant distribution $\mu_i(dz_i)=\frac{1}{q_i \sqrt{\pi}}e^{-z_i^2/q_i^2}dz_i$, it follows from Lemma A.2 in \cite{doi:10.1137/100811738} that $[Z_1(t),...,Z_N(t)]^\top$ is ergodic with invariant distribution $\mu_1\times...\times\mu_N$. Now we define state $z_i=Z_i(t/\vartheta)$ and we let $\overline{\eta}(z)\in\mathbb{R}^N$ be given by \eqref{etabar} but with $z_i$ in place of $Z_i(t/\vartheta)$, i.e
\begin{equation}
    [\overline{\eta}(z)]_k=\begin{dcases}
        a_kf_k(z_k)+\delta_k \sum_{j\in\mathcal{D}_k}a_jf_j(z_j) & \text{if } k \in \mathcal{D} \\
        a_k f_k(z_k) & \text{else}
    \end{dcases}
\end{equation}
Then the system \eqref{dui} becomes
\begin{equation}\label{dui2}
    \dot{u}_i=-\frac{\gamma_i k_i}{a_i}f_i(z_i)\left(J_i(u)+\overline{\eta}(z)^\top\nabla J_i(u)+\frac12\overline{\eta}(z)^\top A_i \overline{\eta}(z) \right)
\end{equation}
We can then use A.4 in \cite{doi:10.1137/100811738} to compute the average system of \eqref{dui2}. Given a state $u$, we use $\tilde{u}$ to denote the corresponding state in the average system. Then, by repeatedly applying Fubini's theorem, we obtain
\begin{align}
    \dot{\tilde{u}}_i&=-\frac{\gamma_i k_i}{a_i}\int_{\mathbb{R}^N}f_i(z_i)\left(J_i(\tilde{u})+\overline{\eta}(z)^\top\nabla J_i(\tilde{u})+\frac12\overline{\eta}(z)^\top A_i \overline{\eta}(z) \right) \mu_1(dz_1)\times ... \times \mu_N(dz_N)\\
    &=-\frac{\gamma_i k_i}{a_i}\int_{\mathbb{R}^N}f_i(z_i)\overline{\eta}(z)^\top \nabla J_i(\tilde{u})\left( \prod_{l=1}^N \frac{1}{q_l \sqrt{\pi}}e^{-\frac{z_l^2}{q_l^2}}\right)dz_1 ... dz_N\\
    &=-\frac{\gamma_i k_i}{a_i}\int_{\mathbb{R}^N} \left(a_i f_i^2(z_i)\nabla_i J_i(\tilde{u})+\sum_{j\in\ck_i}\tilde{\delta}_j a_i f_i^2(z_i)\nabla_j J_i(\tilde{u})\right) \left( \prod_{l=1}^N \frac{1}{q_l \sqrt{\pi}}e^{-\frac{z_l^2}{q_l^2}}\right)dz_1 ... dz_N\\
    &=-k_i\left(\nabla_i J_i(\tilde{u})+\sum_{j\in\ck_i}\tilde{\delta}_j \nabla_j J_i(\tilde{u})\right)\\
    &=-k_i [\overline{\ca}(\tilde{\delta})\tilde{u}+\overline{\cb}(\tilde{\delta})]_i
\end{align}
For the $\delta$ subsystem, we have for $i\in\mathcal{D}$:
\begin{align}
    \tilde{\delta}_i&=\varepsilon\varepsilon_i\int_{\mathbb{R}^N} \left(J_i(\tilde{u}+\overline{\eta}(z))-J_i^\text{ref}\right)\mu_1(dz_1)\times ... \times \mu_N(dz_N)\\
    &=\varepsilon\varepsilon_i\int_{\mathbb{R}^N} \left(J_i(\tilde{u})+\overline{\eta}(z)^\top \nabla J_i(\tilde{u})+\frac12\overline{\eta}(z)^\top A_i \overline{\eta}(z)-J_i^\text{ref}\right)\left( \prod_{l=1}^N \frac{1}{q_l \sqrt{\pi}}e^{-\frac{z_l^2}{q_l^2}}\right)dz_1 ... dz_N\\
    &=\varepsilon\varepsilon_i\left(J_i(\tilde{u})-J_i^\text{ref}+\int_{\mathbb{R}^N}\frac12\overline{\eta}(z)^\top A_i \overline{\eta}(z)\left(\prod_{l=1}^N \frac{1}{q_l \sqrt{\pi}}e^{-\frac{z_l^2}{q_l^2}}\right)dz_1...dz_N\right)\\
    &=\varepsilon\varepsilon_i\left(J_i(\tilde{u})-J_i^\text{ref}+\mathcal{P}_i\left(\begin{bmatrix}
        a\\ \tilde{\delta}
    \end{bmatrix}\right)\right)
\end{align}
where $a=[a_1,...,a_N]$ and $\mathcal{P}_i:\mathbb{R}^{N+n}\to \mathbb{R}$ is a quadratic function satisfying $\mathcal{P}_i\left(\begin{bmatrix}
        \mathbf{0}_N\\ \delta
    \end{bmatrix}\right)$=0. Consider the $\tilde{\zeta}$ system with $a=\mathbf{0}_N$. Since $J^\text{ref}\in\mathcal{J}$, we can apply a standard singular perturbation argument (as done in \cite{tang2024deception}) to conclude that there exists some $\varepsilon^*$ such that for all $\varepsilon\in (0, \varepsilon^*)$, this system has an exponentially stable equilibrium point of the form $\zeta^*=\begin{bmatrix}
        -\overline{\ca}(\delta^*)^{-1}\overline{\cb}(\delta^*)\\
        \delta^*
    \end{bmatrix}$ for some $\delta^*\in\Delta$. By the implicit function theorem, there exists some neighborhood $U$ of $\mathbf{0}_N$ and $\mathcal{C}^1$ function $g:U\to \mathbb{R}^{N+n}$ such that $g(a)$ is an exponentially stable equilibrium point of the $\tilde{\zeta}$ system for all $a\in U$. Moreover, $U$ can be chosen such that
    \begin{equation}
        \left\lvert\mathcal{P}_i\left(\begin{bmatrix}
            a\\ [g(a)]_{N+1}\\\vdots
            \\ [g(a)]_{N+n}
        \end{bmatrix}\right)\right\rvert<\tilde{\varepsilon}\quad \forall i\in\mathcal{D},\ \ a\in U
    \end{equation}
    which establishes item 1) of Theorem \ref{thm}. We can then define $a^*=\sup\{r>0 : \mathcal{B}_r(\mathbf{0}_N)\subset U\}$ and apply Theorem 3 from \cite{doi:10.1137/100811738} to complete the proof.
          \begin{figure}[t!]
  \centering
  \includegraphics[width=0.5\textwidth]{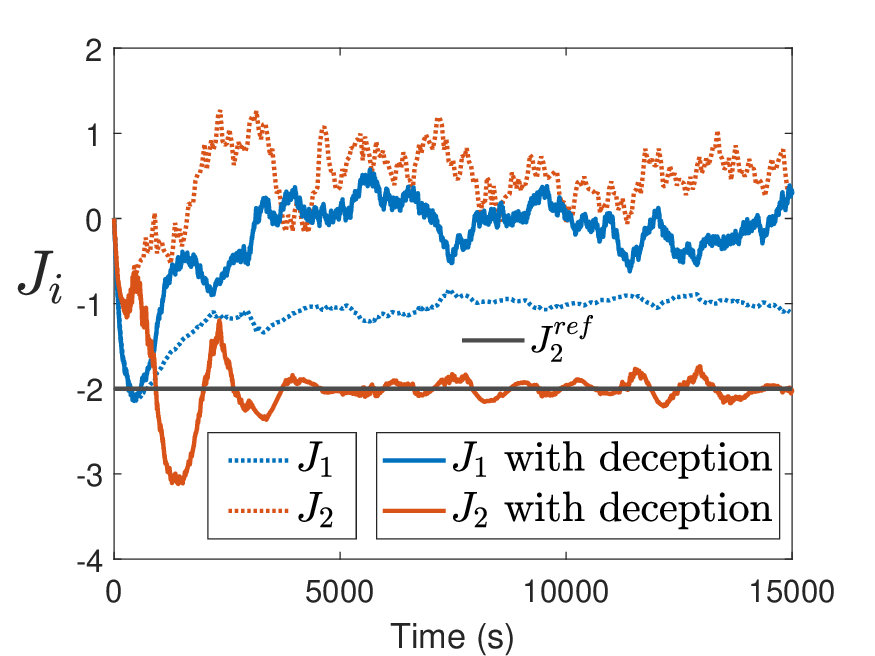}
    \includegraphics[width=0.5\textwidth]{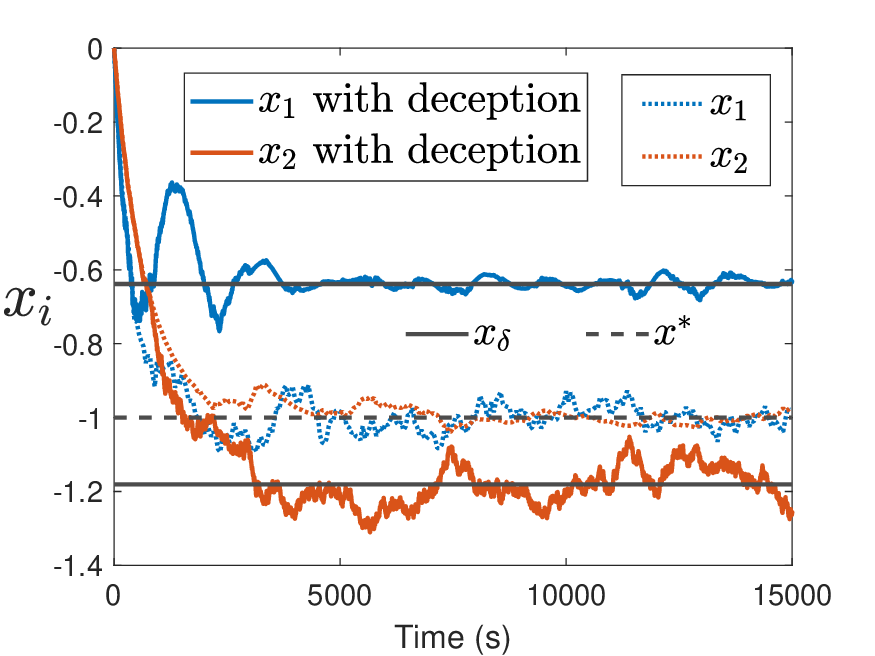}\hspace{-0.5cm}\includegraphics[width=0.5\textwidth]{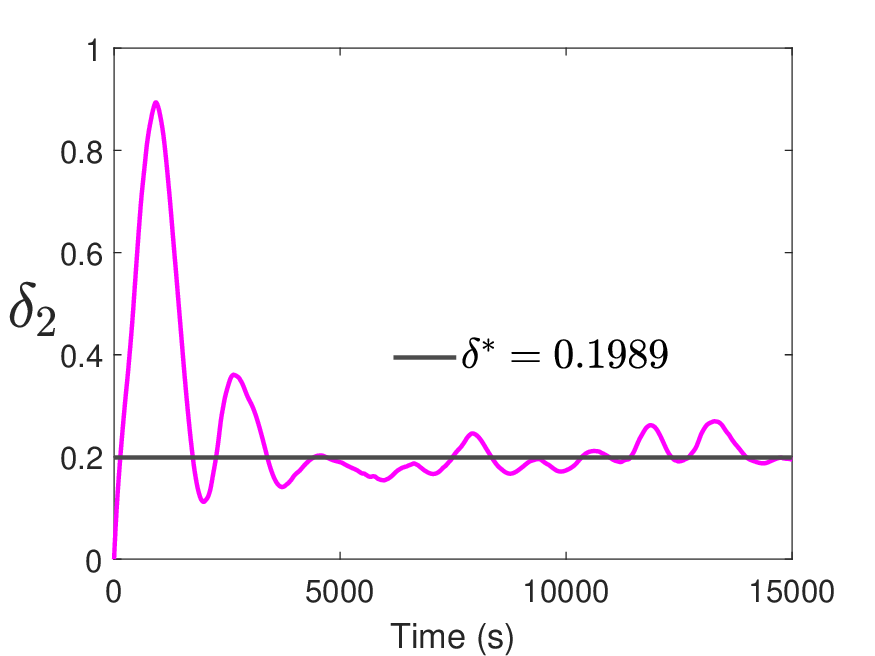}
    \caption{\small Top: Cost function convergence for the two player game from Section \ref{sec_ex} with $J_2^\text{ref}=-2$. Left: The action convergence for the two player game. Right: The corresponding trajectories of $\delta_1$ when player 2 deceives player 1.}\label{figex}
\end{figure}
    \begin{remark}
        One key difference between the above proof and the proof of the main stability result in \cite{tang2024deception} is that the approach taken in \cite{tang2024deception} treats $\mathcal{P}_i$ as an $\mathcal{O}(a)$ perturbation and applies standard robustness results for deterministic systems to establish exponential stability in a semi-global practical sense. Since the conditions of Theorem 3 from \cite{doi:10.1137/100811738} require exponential stability of an equilibrium point of the averaged system, we leverage the implicit function theorem to establish the existence of such a point for $\max_i a_i$ sufficiently small. As the averaged integrator dynamics are still perturbed, we can at best claim that the new equilibrium can be made arbitrarily close to one that achieves each deceptive players' desired payoff.
    \end{remark}
    \section{Numerical Example}\label{sec_ex}
    To further illustrate our results, we present a numerical example. Consider a two player quadratic game with costs of the form \eqref{jquad}, with parameters
    \begin{equation}
        A_1=\begin{bmatrix}
            3 & 1\\ 1& 5
        \end{bmatrix},\quad b_1=\begin{bmatrix}
            4\\2
        \end{bmatrix},\quad A_2=\begin{bmatrix}
            7&2\\2&4
        \end{bmatrix},\quad b_2=\begin{bmatrix}
            1\\6
        \end{bmatrix},\quad c_1=c_2=0.
    \end{equation}
    The pseudogradient of this game is given by 
    \begin{equation}
        \mathcal{G}(x)=\ca x+\cb,\quad\ca=\begin{bmatrix}
            3&1\\ 2&4
        \end{bmatrix},\quad \cb=\begin{bmatrix}
            4\\6
        \end{bmatrix},
    \end{equation}
    where $\ca$ clearly satisfies Assumption \ref{diag}. This game has a unique Nash equilibrium at $x^*=[-1,-1]^\top$ and steady state costs $J_1(x^*)=-1$, $J_2(x^*)=0.5$. For the stochastic NES dynamics \eqref{snes}, the players use $a_1=0.1, k_1=0.06, f_1(s)=\text{sat}(s), a_2=0.12, k_2=0.05, f_2(s)=\sin(s)$ where the saturation function $\text{sat}(\cdot)$ is given by
    \begin{equation}
        \text{sat}(s)=\begin{dcases}
            1 &\text{if } s\ge 1\\
            s &\text{if } s\in (-1, 1)\\
            -1 &\text{if }s\le -1
        \end{dcases}.
    \end{equation}
     The stochastic perturbations $\eta_i(t)$ are generated by \eqref{ou} with $\vartheta_1=0.0005, q_1=0.1, \vartheta_2=0.0004, q_2=0.07$. Now we suppose player 2 is deceptive to player 1, so player 2 implements \eqref{sdnes} and \eqref{ddelta} with $\varepsilon\varepsilon_2=0.001$ and $J_2^\text{ref}=-2$. With Lemma \ref{lemma_delta}, we obtain that $\mathcal{B}_\frac13(0)\subset\Delta$. It can be verified that $J_2^\text{ref}$ is achievable with $\delta^*=0.1989\in\Delta$, which results in a ``deceptive" Nash equilibrium of $x_\delta=-\overline{\ca}(\delta^*)\overline{\cb}(\delta^*)=[-0.6385, -1.1808]^\top$. We also have $\nabla\xi_2(\delta^*)=-12\varepsilon_2<0$ for $\varepsilon_2$, where $\xi_2$ is given in Definition \ref{jattaindef}. Hence, the closed loop stochastic deceptive NES dynamics are stable. The trajectories of the players' actions and costs are provided in Figure \ref{figex}.
    \section{Conclusion}
    In this paper, we introduced stochastic deception in Nash equilibrium seeking problems for noncooperative games with finitely many players. We leverage results on multi-input stochastic averaging to establish local exponential (in probability) convergence for the deceptive stochastic NES dynamics. We primarily consider quadratic games, as this allows us to consider a more general class of NES dynamics with an improved estimate of the ``stability preserving set" $\Delta$, but future directions will focus on extending the results to general nonlinear payoffs. Since our deception mechanism requires real-time measurements of the victim's stochastic perturbation, it is also of future interest to study methods a deceptive player could use to obtain such measurements or information in real time.

\bibliography{mybibliography}

\end{document}